 \documentclass[smallabstract,smallcaptions]{dccpaper}

\usepackage{epsfig}
\usepackage{citesort}
\usepackage{amsmath}
\usepackage{amssymb}
\usepackage{color}
\usepackage{url}

\usepackage{multirow}

\usepackage{xspace}

\usepackage{amsthm}

\usepackage{xparse}

\usepackage{comment}
\usepackage{color}

\definecolor{darkgreen}{rgb}{0.0,0.7,0.0}
\newenvironment{aj}{\noindent\color{magenta} Artur:} {}

\newenvironment{mg}{\noindent\color{blue} Micha?:} {}

\usepackage{siunitx}
\newcommand{\cnt}[1]{ \text{\#}(#1) }

\newcommand{\procfont}[1]{\textnormal{\textsf{#1}}}

\providecommand{\Ocomp}{{\mathcal O}}

\newtheorem{lemma}{Lemma}
\newtheorem{theorem}{Theorem}
\newtheorem{definition}{Definition}
\newtheorem{fact}{Fact}

\newlength{\figurewidth}
\newlength{\smallfigurewidth}

\begin{document}

\title
{\large
\textbf{Towards Better Compressed Representations \footnote{Work supported under National Science Centre, Poland project number 2017/26/E/ST6/00191.}
}
}

\author{%
	Micha\l{} Ga\'nczorz$^{\ast}$\\[0.5em]
	{\small\begin{minipage}{\linewidth}\begin{center}
				\begin{tabular}{c}
					$^{\ast}$University of Wroc{\l}aw,
					Institute of Computer Science\\
					ul.~Joliot-Curie~15\\
					Wroc{\l}aw, PL50383, Poland\\
					\url{mga@cs.uni.wroc.pl} \\
				\end{tabular}
			\end{center}\end{minipage}}
		}

\maketitle
\thispagestyle{empty}

\vspace{5pt}
\begin{abstract}
We introduce the problem of computing 
a parsing where each phrase is of length at most $m$ and which minimizes the zeroth
order entropy of parsing.
Based on the recent theoretical results we devise a heuristic for 
this problem.
The solution has straightforward application in succinct text representations
and gives practical improvements.
Moreover, the proposed heuristic yields structure whose size can be bounded
both by $|S|H_{m-1}(S)$ and by $|S|/m(H_0(S) + \cdots + H_{m-1}(S)$,
where $H_{k}(S)$ is the $k$-th order empirical entropy of $S$.
We also consider a similar problem in which the 
first-order entropy is minimized.
\end{abstract}

\Section{Introduction}
As the amount stored data grows exponentially,
over the last decade we have seen a rapid growth of importance
of compressed data structures, both in practical applications and in theoretical research.
In this paper we take a closer look at one of the simplest compressed data structure,
which solves the static random access problem,
namely given a text $S$ builds its compressed representation
supporting only one operation: \procfont{access}($i$) which returns $i$-th element,~$S[i]$.
This structure is already useful in many applications, such as databases or compressed~RAM.

The two main paradigms in lossless data compression are:
dictionary compression, e.g.\ LZ77 or grammar compression, and entropy-based compression, e.g.\ PPM.
It is generally acknowledged that the former works better for highly repetitive data,
while the latter for data which does not have many repeating substrings.

There are many theoretical results for compressed representations using the entropy paradigm~\cite{sadakane2006squeezing,FerraginaV07SimpStat,GonzalezNStatistical,FerraginaMMN07CompressedRepresentation},
there are even structures that allow not only to read but also to modify the data~\cite{NavarroOptimalDynamic,GrossiDynamicIndexes,JanssonCRAM}.
For the static random access problem,
those structures achieve the following bounds:
for a string $S$ over the alphabet of size $\sigma$
the space is bounded by
$|S|H_k(S) + \Ocomp{ \left( \frac{|S|}{\log_\sigma |S|}\cdot 
	\left(k \log \sigma + \log \log |S| \right) \right) }$
and the \procfont{access} query takes $\Ocomp(1)$ time.
It is also worth noting that some of these structures also allow to read 
$\Ocomp(\log_\sigma |S|)$ bits in $\Ocomp(1)$ time, this is useful in sequential read,
which is often used in RAM.
What is somehow surprising, is that the analyses of space usage of those structures
are usually done in the same way:
first, it is shown that a structure induces a parsing of the input string,
(sometimes the parsing is explicitly constructed in the structure,
sometimes it is defined implicitly);
then it is shown that the size of data structure is dominated by  the zeroth order entropy of this parsing, i.e.\ zeroth order entropy of string made by replacing parsing phrases with new letters;
lastly, this zeroth-order entropy is related to the $k$-th order entropy of input string $S$.

%

A recently shown theorem strengthens and generalizes
previous estimations of parsing's zeroth-order entropy in terms of input's $k$-th order entropy~\cite{FerraginaV07SimpStat,KosarajuManzini99}.

\begin{theorem}[{\cite[Theorem~1]{entropyBoundsUnpublished}}]
\label{thm:thm1}
Let $S$, be a string, $Y_S$ its parsing, $n = |S|$. Then: 
$$
|Y_S| H_0(Y_S)
\leq 
|S|H_k(S) + |Y_S|k\log \sigma + |L| H_0(L),
$$
where $L$ is a string over alphabet $1\ldots |Y_S|$ whose letters are lengths of factors of $Y_S$.
Moreover, if $|Y_S| = \Ocomp(|S|/\log_\sigma |S|)$
then this bound becomes:
$$
|Y_S| H_0(Y_S)
\leq 
|S|H_k(S) + \Ocomp{ \left( \frac{|S|}{\log_\sigma |S|}\cdot 
	\left(k \log \sigma + \log \log |S| \right) \right) }.
$$
\end{theorem}
The aforementioned structures construct a parsing of $S$ into 
$\Ocomp(|S|/\log_\sigma|S|)$ phrases of fixed length;
most of them find a naive parsing into phrases
of length $\Theta(\log_\sigma |S|)$.
For example, the simplest one~\cite{FerraginaV07SimpStat}
first finds a parsing into equal-length phrases of length $\frac{1}{2} \log_\sigma |S|$,
assigns each phrase a prefix-free code and concatenates the codes.

Theorem~\ref{thm:thm1} shows that using Huffman codes for any reasonable parsing
yields a~data structure of size $|S|H_k(S)$ (plus the dictionary, which is often small).
It is natural to ask, can we get better upper bounds if we can \emph{choose} the parsing,
it turns out that by recent result, this is indeed the case:
\begin{theorem}[{\cite[Theorem~2]{entropyBoundsUnpublished}}]
	\label{thm:mean_entropy}
	Let $S$ be a string over alphabet $\sigma$.
	Then for any integer $l$
	we can construct a parsing $Y_S$ of size $|Y_S| \leq
	\left\lceil\frac{|S|}{l}\right\rceil + 1$ satisfying:
	$$
	|Y_S| H_0(Y_S)
	\leq
	\frac{|S|}{l} \sum_{i=0}^{l-1} H_i(S) + \Ocomp(\log |S|).
	$$
	All phrases except the first and last one have length $l$,
	and all phrases lengths are bounded by $l$.
\end{theorem}

As $\log \sigma \geq H_i(S) \geq H_j(S)$ for $i \leq  j$,
bound from Theorem~\ref{thm:mean_entropy} is smaller than the one from Theorem~\ref{thm:thm1}.
Moreover, in practice entropies tend to get significantly smaller with $i$.
Observe that there are at most $l-1$ parsings satisfying conditions of Theorem~\ref{thm:mean_entropy},
so it basically says that one of them satisfies the improved bound.

While Theorem~\ref{thm:mean_entropy} gives us meaningful insight that
some parsings can beat the bound of Theorem~\ref{thm:thm1},
it does not account for significant practical gains.
What we are really interested in is a parsing which
will have small entropy and will be (practically) useful
in structures for compressed representation.

The first observation is that existing structures can be generalized
so that they support parsings whose phrases are of length at most 
$m=\Ocomp(\log_\sigma|S|)$,
call such parsing $m$-bounded,
instead of requiring them to be equal.
This motivates the following problem:

\begin{definition}[Minimum Entropy Bounded-Factor Parsing Problem]
\label{def:problem_zero}
Given an integer $m$ and a string $S$ compute its $m$-bounded parsing $Y_S$
minimizing $|Y_S|H_0(Y_S)$ over all $m$-bounded parsings.
\end{definition}

Computing the optimal solution to this problem seems difficult,
as entropy is a~global measure and decisions concerning the parsing cannot be done locally;
moreover, entropy minimization problems tend to be hard.
We set a more realistic goal of finding an efficient and simple heuristic.

The main idea of the heuristic comes from the proofs of Theorem~\ref{thm:thm1} and \ref{thm:mean_entropy}:
as computing the exact entropy is difficult, a simpler to compute upper bound $H'(Y_S)\geq H_0(Y_S)$ on the entropy of the parsing.
Then instead of trying to compute a parsing minimizing $H(Y_S)$,
we compute one minimizing $H'(Y_S)$.
The definition of $H'$ is simple enough so that computing it can be done in almost linear time.

The main drawback of aforementioned structures is that the additional
factor $\Ocomp{ \left( |S|
	\left(k \log \sigma + \log \log |S| \right)/\log_\sigma |S| \right) }$
grows significantly with $k$.
As a solution, Grossi et al.~\cite{GrossiDynamicIndexes}
encoded the parsing with a first order entropy coder;
the size of such data structure is at most
$|S|H_k(S) + \Ocomp\left(|S|\log \log_\sigma |S|/\log_\sigma|S|\right)$.
In fact, any parsing of size $\Ocomp(|S|/\log_\sigma |S|)$ and which factors are of length 
at least $k$ satisfies this bound~\cite{GrossiDynamicIndexes}.

Even though neither of~\cite{GrossiDynamicIndexes,entropyBoundsUnpublished} give explicit
bounds for the parsing encoded with first order entropy,
it is easy to apply methods
from~\cite{GrossiDynamicIndexes,entropyBoundsUnpublished} and obtain the
versions of Theorem~\ref{thm:thm1} and Theorem~\ref{thm:mean_entropy}:
\begin{theorem}
	\label{thm:thm2}
	Let $S$, $|S|=n$ be a string, $Y_S = y_1 y_2 \cdots y_{|Y_S|}$ its parsing,
	where $|y_i| \geq k$, for all $i$.
	Then: 
	$$
	|Y_S| H_1(Y_S)
	\leq 
	|S|H_k(S) + |L| H_0(L),
	$$
	where $L$ is a string over alphabet $1\ldots |Y_S|$ which letters are lengths of factors of $Y_S$.
	Moreover, if $|Y_S| = \Ocomp(|S|/\log_\sigma |S|)$
	then this bound becomes:
	$$
	|Y_S| H_1(Y_S)
	\leq 
	|S|H_k(S) + \Ocomp{ \left( \frac{|S|\cdot \log \log |S|}{\log_\sigma |S|} \right) }.
	$$
\end{theorem}

\begin{theorem}
	\label{thm:mean_entropy_2}
	Let $S$ be a string over alphabet $\sigma$.
	Then for any integer $l$
	we can construct a parsing $Y_S$ of size $|Y_S| \leq
	\left\lceil\frac{|S|}{l}\right\rceil + 1$ satisfying:
	$$
	|Y_S| H_1(Y_S)
	\leq
	\frac{|S|}{l} \sum_{i=l}^{2 l-1} H_i(S) + \Ocomp(\log |S|).
	$$
	All phrases except the first and last one have length $l$,
	and all phrases lengths are bounded by $l$.
\end{theorem}

\noindent This motivates a generalization of 
minimum entropy bounded-factor parsing problem:
\begin{definition}[Minimum First-Order Entropy Bounded-Factor Parsing Problem]
\label{def:problem_first}
Given an integer $m$ and a string $S$ compute its $m$-bounded parsing $Y_S$
minimizing $|Y_S|H_1(Y_S)$ over all $m$-bounded parsings.
\end{definition}
We show that our heuristic can be generalized to this problem.

\Section{A Better Parsing}
Theorem~\ref{thm:thm1} and~\ref{thm:mean_entropy} utilize the following Lemma to upper-bound the entropy of a parsing:
\begin{lemma}[\cite{aczel1973shannon}]\label{lemma_prob}
	Let $w$ be a string over alphabet $\Gamma$
	and $p:\Gamma \rightarrow \mathbb{R}^+$ be a function such that $\sum_{s \in \Gamma} p(s)  \leq 1$.
	Then:
	$$
	|w|H_0(w) \leq -\sum_{s \in \Gamma} |w|_{s} \log p(s)\enspace .
	$$
\end{lemma}
Lemma~\ref{lemma_prob} should be understood as follows:
we can assign each different phrase $y$ a value $p(y)$,
then each occurrence of $y$ will contribute $-\log p(y)$
to the ``entropy'' of parsing;
the assumption that values $p$ for different phrases sum up to at most $1$
is needed to ensure that $p$ behaves like a probability distributions on phrases.
In fact, for such defined function $p$
we can compute a prefix-free coding that assigns a code of length roughly
$-\log p(y)$ to $y$, e.g.\ both Huffman and arithmetic coding
can be used to obtain the codes of length $-\log p(y) + \Ocomp(1)$.

The actual functions $p$ utilised by
Theorem~\ref{thm:thm1} and~\ref{thm:mean_entropy} are as follows:
For a factor $y=a_1 \cdots a_j$,
define $p(y) = p_{\text{len}}(|y|) \cdot \prod_{i=1}^j p_i(a_i)$,
where $p_{\text{len}}(|y|)$ is the empirical probability that a factor has a length $|y|$, i.e.\ $\cnt{|y|}/|Y|$,
where $\cnt{|y|}$ is number of factors in $Y$ of length $|y|$.
Theorem~\ref{thm:thm1} and~\ref{thm:mean_entropy} use different $p_i$s,
for the former it~is~\footnotemark[1]:
\begin{align*}
p_i(a) &=
\begin{cases}
\frac{1}{\sigma} & \text{for  $i \leq k$, where $\sigma$ is the size of the alphabet},\\
p_i(a_i) = \frac{\cnt{a_{i-k}\cdots a_i}}{\cnt{a_{i-k}\cdots a_{i-1}}}, &\text{for } i > k,
\end{cases}
\intertext{for the latter it is~\footnotemark[1]:}
p_i(a_i) &= \frac{\cnt{a_{1}\cdots a_i}}{\cnt{a_{1}\cdots a_{i-1}}},
\end{align*}
where $\cnt{w}$ is a number of occurrences of world $w$ in the input string $S$.
\footnotetext[1]{ If $a_{i-k}\cdots a_i$ or $a_{1}\cdots a_i$ is the suffix of $S$
	the values should be $p_i(a_i) = \frac{\cnt{a_{i-k}\cdots a_i}}{\cnt{a_{i-k}\cdots a_{i-1}}-1}$
and $p_i(a_i) = \frac{\cnt{a_{1}\cdots a_i}}{\cnt{a_{1}\cdots a_{i-1}}-1}$ respectively,
yet as this does not change $\sum_{y_i \in Y_S} - \log p_l(|y_i|)$ significantly
we use those values as they simplify both the reasoning and algorithm.}

The sum, over all different phrases, of such defined values is at most $1$~\cite{entropyBoundsUnpublished},
hence satisfying conditions of Lemma~\ref{lemma_prob}.
Observe that these functions estimate not only the cost of entropy coding
of a phrase,
but also such cost for each individual letter in phrase:
given a phrase $y= a_1 a_2 \cdots a_j$
in the first case for $i \leq k $ the cost of encoding a~letter $a_i$ is $-\log p_i(a_i)$,
which corresponds to the naive encoding with $\log \sigma$ bits,
and for $i > k$ the cost corresponds to encoding with $k$-th order entropy coder.
Similarly, in the second case, the cost of encoding the letter,
$-\log p_i(a_i)$,
corresponds to the cost of encoding letter $a_i$ with $i-1$-th order entropy coder
(i.e.\ first letter is encoded with $0$-th order entropy coder, second with a $1$-st order coder and so on).
We also note that $\sum_{y_i \in Y_S} - \log p_l(|y_i|)$ sums up to the entropy of phrases' lengths, i.e.\ $|L|H_0(L)$.


In order to use some function $p$ in a heuristic
$p$ should depend only on the string on not on the parsing itself,
then dynamic programming can be used to compute a~parsing $Y_S$ minimizing $\sum_{y_i \in Y_S} - \log p(y_i)$.
Unfortunately, 
$p_{\text{len}}$ depends on the parsing. 
We thus modify $p_{\text{len}}$ and set it to be $1/m$:
we are interested in factors of length at most $m$, so
$p_{\text{len}}$ is a probability distribution on $\{1, \ldots, m\}$.
Such modification should have little effect on the encoding size---we are interested in a parsing with short factors,
hence the factor corresponding to the entropy of lengths
$|L|H_0(L)$ should be small.

We will use a variant of $p_\sigma$ of Theorem~\ref{thm:mean_entropy},
as it carries more information on the string structure that the other one
(and gives significantly better practical results).

\begin{definition}
\label{def:pvalues_zero}
Given a string $S$ and integer $m$ denoting the maximum substring length,
we define $p$ for a substring $y = a_1 a_2 \cdots a_j$ of $S$ as
$$
p_{H_0}(y) = \frac 1 m \cdot p_1(a_1) \cdot p_2(a_2)\cdots p_j(a_j),
\text{ where }
p_i(a_i) = \frac{\cnt{a_1 a_2 \cdots a_i}}{\cnt{a_1 a_2 \cdots a_{i-1}}}
$$
and phrase cost of $y$ as: $-\log p_{H_0}(y)$.
\end{definition}
\begin{fact}
\label{fact:p_telescope_1}
For a phrase $y$, $p_{H_0}(y) = \frac{1}{m}\cdot \frac{\cnt{y}}{|S|}$
\end{fact}
\noindent 
We can now give a heuristic for minimum entropy bounded-factor parsing problem:
\begin{lemma}
\label{lem:algo_zero}
For a string $S$ we can compute in $\Ocomp(|S|\cdot m)$ time
an $m$-bounded parsing which minimizes the sum, over all phrases, of values~$p_{H_0}$.
\end{lemma}
\begin{proof}
We apply standard dynamic programming, let the cost of a parsing be the sum of $p$
over all its phrases.
Let $dp(i)$ denote the smallest cost of parsing of $S[1\ldots i]$.
Assuming we computed $dp(j)$ for $j < i$ we can compute $dp(i)$ using:
$dp(i) = \min\{dp(i-j) + p(S[i-j\ldots i]) \ | \  1\leq j \leq m \}$.\\
Computing each $dp(i)$ takes $\Ocomp(m)$ time,
assuming we can access $p(\cdot)$ value in $\Ocomp(1)$ time.
To this end we preprocess the input to get the number of occurrences of $y$ in $S$,
e.g.\ by constructing suffix tree with appropriate structures,
and use Fact~\ref{fact:p_telescope_1}.
We retrieve the parsing by backtracking. 
\end{proof}

Observe that the algorithm finds a parsing whose cost with respect to the
function $p(\cdot)$ is not greater than both of the estimates used in proof
(not considering the $|L|H_0(L)$),
thus by Lemma~\ref{lemma_prob}, we obtain the following bounds:
\begin{lemma}
\label{lem:alg_zero_bound}
For a~given string $S$ over alphabet of size $\sigma$
and a parameter $m$
algorithm from Lemma~\ref{lem:algo_zero} finds a parsing $Y_S$ of $S$ such that
both the following inequalities holds:
\begin{align*}
|Y_S|H_0(Y_S) &\leq \frac{|S|}{m} \sum_{0 \leq i \leq m-1} H_i(S) + |Y_S|\log m\\
|Y_S|H_0(Y_S) &\leq |S|H_{k}(S) + |Y_S|\log m + |Y_S| k \log \sigma, & \text{ for each} 0 \leq k < m.
\end{align*}
\end{lemma}

In the following we develop a heuristic for 
minimum first-order entropy bounded-factor parsing.
To this end we analyze the main idea of proofs of Theorem~\ref{thm:mean_entropy_2} and~\ref{thm:thm2}:
we want to assign each phrase a probability $p$ (or a prefix-free code)
and apply Lemma~\ref{lemma_prob}.
However, Lemma~\ref{lemma_prob} works only for the zeroth order entropy.
Still, $H_1$ is defined through $H_0$ of appropriate strings:
for a string $T \in \Sigma^*$,
$|T|H_1(T) =  \sum_{T_\sigma, \sigma \in \Sigma} |T_\sigma| H_0(T_\sigma)$,
where $T_\sigma$ is a string made by concatenating all letters of $T$
which occur in one-letter context $\sigma$ (e.g.\ for $T=abacaac$, $T_a = bcac$).
Thus, for a parsing $Y_S$, we construct $T_y$ for each different phrase $y$ of a parsing and apply Lemma~\ref{lemma_prob} to it.

The main reason, why we can obtain better (theoretical and practical)
bounds for $H_1$ of a parsing, is that
for a given factor $y = a_1 a_2 a_3 \cdots a_{|y|}$ from $T_{y'}$
we know that $y'$ precedes $y$, thus we can include $y'$ in a context.
As a result, in the definition of $p(\cdot) = \frac{1}{m}\cdot p_1(a_1) \cdot p_2(a_2) \cdots p_{|y|}(p_{|y|})$
we define $p_i(a_i)$ as empirical probability of $a_i$
in $|y'| + i - 1$-letter context,
i.e.\ $p_\sigma(a_1 a_2 \cdots a_i) = \cnt{y' a_1 a_2 \cdots a_i }/\cnt{y'}$.
This corresponds to the first letter of the factor being encoded with $|y'|$-order entropy coder,
second with $|y'|+1$-order entropy coder and so on.

In the $H_1$ variant our algorithm we use the $p(\cdot)$ values
used to prove the Theorem~\ref{thm:mean_entropy_2}:
\begin{definition}
	\label{def:pvalues_first}
	Given a string $S$, integer $m$ denoting the maximum substring length
	and a substring $y'y$, where  $y' = b_1 b_2 \cdots b_h$, $y = a_1 a_2 \cdots a_j$,
	we define value $p_{H_1}$ for a~substring $y = a_1 a_2 \cdots a_j$ of $S$:\\
	$$
	p_{H_1}(y, y') = \frac 1 m \cdot p_1(a_1) \cdot p_2(a_2)\cdot \cdots \cdot p_j(a_j)
	\text{, where } 
		p_i(a_i) = \frac{ \cnt{y' a_{1}\cdots a_i}}{ \cnt{y' a_{1}\cdots a_{i-1}}},
	$$
	and phrase cost of $y$ preceded by phrase $y'$ as $-\log p_{H_1}(y, y')$
\end{definition}
Such defined values $p_{H_1}(\cdot)$ are still easy to compute:
\begin{fact}
	\label{fact:p_telescope_2}
	For a phrase $y$ and phrase $y'$ preceding $y$,
	$p_{H_1}(y, y') = \frac{1}{m}\cdot \frac{\cnt{y'y}}{\cnt{y'}}$
\end{fact}

We can extend our dynamic programming so that it computes the optimal parsings
with respect to the phrase cost defined above;
note that the time complexity increases, as now 
the phrase cost is dependent on the previous phrase cost,
hence we loop not only on the possible phrase lengths,
but also on the lengths of the previous phrase.

\begin{lemma}
	\label{lem:algo_first}
	For a string $S$ we can compute in $\Ocomp(|S|\cdot m^2)$ time
	an $m$-bounded parsing minimizing the sum, over all phrases,~of~values~$p_{H_1}$.
\end{lemma}
\begin{proof}
	Let $dp(i, u)$ denote the smallest cost of computing a parsing of $S[1\ldots i]$ with the last factor of length $u$.
	Assuming we computed $dp(j, v)$ for $j < i$, we can compute $dp(i, u)$ follows:
	$$
	dp(i, u) = \min\{dp(i-u, v) + p_{H_1}(S[i-u+1\ldots i], \ S[i-v-u+1 \ldots i-u]) \ | \  \ 1\leq v \leq m \}.
	$$
	Smallest cost is equal to $\min \{dp(|S|, u), 1\leq u\leq m \}$.
	We retrieve the parsing by backtracking.
\end{proof}

Again, we have theoretical guarantees on the size of the computed parsing:
\begin{lemma}
	For a given string $S$ over alphabet of size $\sigma$
	and a parameter $m$
	algorithm from Lemma~\ref{lem:algo_first} finds a parsing $Y_S$ of $S$ such that
	both the following inequalities holds:
\begin{align*}
	|Y_S|H_1(Y_S) &\leq \frac{|S|}{m} \sum_{ m \leq i \leq 2m-1} H_i(S) + |Y_S|\log m\\
	|Y_S|H_1(Y_S) &\leq |S|H_{m}(S) + |Y_S|\log m.
\end{align*}
\end{lemma}

\Section{Experimental results --- Entropy Comparison}
In the next sections we present the experimental results.
Our implementation make use of sdsl library~\cite{gbmp2014sea}.
The implementation is available at~\url{https://github.com/iguana-ben/compressed-representation}.
Test data is from \textit{Pizza \& Chilli} corpus~\url{http://pizzachili.dcc.uchile.cl/texts.html}.


Table~\ref{tab:zero_entropy_compare} contains the comparison of the parsing returned by
$H_0$ version of our parsing algorithm (denoted by $A$) 
compared to parsing obtained by applying Theorem~\ref{thm:mean_entropy} (denoted by $B$),
i.e.\ by evaluating $m$ naive parsings and returning the one with smallest entropy.
In general, both the entropy of parsing and number of different phrases
(denoted by $|\Sigma|$) are smaller for the algorithm's parsing:
while the $|A|H_0(A)$ is smaller than $|B|H_0(B)$ by a couple percent,
the number of different phrases is reduced significantly, even halved.
Note that $|\Sigma|$ corresponds closely to the encoding size of dictionary
(e.g.\ Huffman) when using
zeroth order entropy coder: on average $\Theta\left( |\Sigma|\log |\Sigma|\right)$ bits are both sufficient and necessary for a~random sequence over alphabet of size $|\Sigma|$;
moreover with arithmetic coding we still need to store
$|\Sigma|$ frequencies, which gives similar bound.

Table~\ref{tab:zero_entropy_compare} also shows an interesting phenomenon:
for small $m$ Theorem~\ref{thm:mean_entropy} seems to be tight.
Also, the \textit{dna} sequence was the only one on which there was no improvement;
note that in this case first $m$ entropies are almost equal,
so the overall text seems random for small $m$,
thus we cannot get much information from $p_{H_0}$
(and we lose on encoding entropy of lengths).

Table~\ref{tab:first_entropy_compare} contains the comparison
of the parsing returned by $H_1$ version
of our parsing algorithm (denoted by $A$)
to the parsing from Theorem~\ref{thm:mean_entropy_2} (denoted by $B$).
Again, we see decrease of entropy, however this time the
number of different phrases, $|\Sigma_A|$, is larger.
This is not so important, though,
as for the first order entropy coder the encoding of the dictionary
closely corresponds to the number of different two-letter words,
as for each letter we build a separate dictionary and store which
letters belong to it,
hence we store $|\text{pairs}(A)|$ numbers,
similarly in the first order arithmetic coding we
store $|\text{pairs}(A)|$ frequencies.
Hence, the number  $|\text{pairs}(A)| \log |\Sigma_{A}|$ seems to better reflect 
the encoding size of the dictionary.
And while $|\Sigma_{A}|$ is larger than $|\Sigma_{B}|$,
$|\text{pairs}(A)|$ is smaller than $|\text{pairs}(B)|$.
Overall, the $|A|H_0(A)$ can be about $5-20\%$ smaller than the $|B|H_0(B)$,
and the $|\text{pairs}(A)|$ can be about $10\%$ smaller than $|\text{pairs}(B)|$.
\begin{table}[h!]
	\begin{center}
		{
			\renewcommand{\baselinestretch}{1}\footnotesize
			\begin{tabular}{|c|c|c|ccc|ccc}
				File & $m$ & $\frac{1}{m} \sum\limits _{i < m} H_i(S) $
				& $\frac{|B|}{|S|}H_0(B)$ & $\frac{|S|}{|B|}$ & $|\Sigma_{B}|$ 
				& $\frac{|A|}{|S|}H_0(A)$ & $\frac{|S|}{|A|}$ & $|\Sigma_{A}|$ \\
				\hline
				\multirow{4}{*}{\textit{english}}
				 & 2 & 4.0677 & 
				 	4.0676 & 2. &  $5.3 \cdot 10^3$ & 
				 	4.0610 & 1.831 & $3.1 \cdot 10^3$ \\
				 & 4 & 3.3607 &
				 	3.3570 & 4. & $2.0 \cdot 10^5$ & 
				 	3.1928 & 3.629 & $.93 \cdot 10^5$ \\
				 & 6 & 2.8698 &
				 	2.8431  & 6. & $1.1 \cdot 10^6$ & 
				 	2.6457  & 5.447 & $.44 \cdot 10^6$ \\
				 & 8 & 2.5166 &
				 	2.4383 & 8. & $2.3 \cdot 10^6$ & 
				 	2.2773 & 7.302 & $.97 \cdot 10^6$ \\
				 \hline
				 \multirow{4}{*}{\textit{dblp.xml}}
				 & 2 & 4.2616 &
				 	4.2615 & 2. & $5.7 \cdot 10^3$ &
				 	4.1718 & 1.849 & $2.5 \cdot 10^3$ \\
				 & 4 & 2.9622 & 
				 	2.9569 & 4. & $2.9 \cdot 10^5$ &
				 	2.8095 & 3.759 & $1.4 \cdot 10^5$ \\
				 & 6 & 2.2562 &
					2.2328 & 6. & $8.3 \cdot 10^5$ &
					2.1294 & 5.775 & $4.9 \cdot 10^5$ \\
				 & 8 & 1.8431 &
				 	1.8025 & 8. & $1.1\cdot 10^6$ &
				 	1.6911 & 7.735 & $.73 \cdot 10^6$\\
				\hline
				\multirow{4}{*}{\textit{sources}}
				& 2 & 4.7877 &
				4.7866 & 2. & $8.1 \cdot 10^3$ &
				4.7230 & 1.826 & $3.0 \cdot 10^3$ \\
				& 4 & 3.7004 & 
				3.6894 & 4. & $5.4 \cdot 10^5$ &
				3.4833 & 3.613 & $2.1 \cdot 10^5$ \\
				& 6 & 2.9813 &
				2.9341 & 6. & $1.7 \cdot 10^6$ &
				2.7636 & 5.504 & $.81 \cdot 10^6$ \\
				& 8 & 2.4884 &
				2.3994 & 8. & $2.5\cdot 10^6$ &
				2.2843 & 7.450 & $1.4 \cdot 10^6$\\
				\hline
				\multirow{4}{*}{\textit{dna}}
				& 2 & 1.9584 &
				1.9584 & 2. & $1.0 \cdot 10^2$ &
				2.0317 & 1.941 & $.87 \cdot 10^2$ \\
				& 4 & 1.9402 & 
				1.9402 & 4. & $7.4 \cdot 10^2$ &
				2.0204 & 3.837 & $6.6 \cdot 10^2$ \\
				& 6 & 1.9295 &
				1.9294 & 6. & $4.7 \cdot 10^3$ &
				2.0087 & 5.699 & $5.5 \cdot 10^3$ \\
				& 8 & 1.9158 &
				1.9149 & 8. & $6.6\cdot 10^4$ &
				1.9899 & 7.501 & $7.1 \cdot 10^4$\\
				\hline
				\multirow{4}{*}{\textit{pitches}}
				& 2 & 5.1835 &
				5.1833 & 2. & $1.0 \cdot 10^4$ &
				5.2284 & 1.786 & $.43 \cdot 10^4$ \\
				& 4 & 4.4908 & 
				4.4438 & 4. & $2.0 \cdot 10^6$&
				4.4235 & 3.504 & $.95 \cdot 10^6$ \\
				& 6 & 3.5928 &
				3.4137 & 6. & $4.6 \cdot 10^6$ &
				3.3805 & 5.680 & $3.4 \cdot 10^6$ \\
				& 8 & 2.8309 &
				2.6216 & 8. & $4.4\cdot 10^6$ &
				2.5949 & 7.823 & $3.7 \cdot 10^6$ \\
				\hline
				\multirow{4}{*}{\textit{proteins}}
				& 2 & 4.1840 &
				4.1840 & 2. & $5.7 \cdot 10^2$ &
				4.2448 & 1.967 & $5.0 \cdot 10^2$ \\
				& 4 & 4.1369 & 
				4.1340 & 4. & $1.9 \cdot 10^5$ &
				4.2075 & 3.793 & $1.7 \cdot 10^5$ \\
				& 6 & 3.8301 &
				3.6291 & 6. & $5.8 \cdot 10^6$ &
				3.6772 & 5.635 & $4.2 \cdot 10^6$ \\
				& 8 & 3.0066 &
				2.7316 & 8. & $5.4 \cdot 10^6$ &
				2.7257 & 7.948 & $5.1 \cdot 10^6$\\
		\end{tabular}
		}
	\caption{\label{tab:zero_entropy_compare}
		Entropy comparison for $H_0$ variant, entropy values are divided by $|S|$ to get the bps,
		all files of size 50MB, $A$ denotes parsing generated by our algorithm,
		while $B$ is obtained by application of Theorem~\ref{thm:mean_entropy}
		(i.e.\ parsing minimizing $|B|H_0(B)$ among $m$ naive parsings).  }
	\end{center}
\end{table}
\hspace{-1cm}
\begin{table}[h!]
	\begin{center}
		{	
			\renewcommand{\baselinestretch}{1}\footnotesize
			\setlength{\tabcolsep}{4pt}
			\begin{tabular}{|c|c|c|cccc|cccc}
				File & $m$ & $\sum\limits _{m\leq i < 2m} \frac{ H_i(S) }{m} $
				& $\frac{|B|}{|S|}H_1(B)$ & $\frac{|S|}{|B|}$ & $|\Sigma_{B}|$ & $|\text{pairs}(B)|$
				& $\frac{|A|}{|S|}H_1(A)$ & $\frac{|S|}{|A|}$ & $|\Sigma_{A}|$ & $|\text{pairs}(A)|$\\
				\hline
				\multirow{3}{*}{\textit{english}}
				& 2 & 2.6537 &
					2.6510 & 2. & $5.34 \cdot 10^3$ &  $2.54 \cdot 10^5$ &
					2.6286 & 1.89 & $5.64 \cdot 10^3$ & $ 2.03 \cdot 10^5$ \\
				& 3 & 2.0540 &
					2.0285 & 3. & $4.56 \cdot 10^4$ &  $1.53 \cdot 10^6$ &
					1.9819 & 2.85 & $5.02 \cdot 10^4$ & $ 1.29 \cdot 10^6$ \\
				& 4 & 1.6726 &
					1.5882 & 4. & $2.01 \cdot 10^5$ &  $3.64 \cdot 10^6$ &
					1.5161 & 3.83 & $2.27 \cdot 10^5$ & $ 3.23 \cdot 10^6$ \\
				\hline
				\multirow{3}{*}{\textit{dblp.xml}}
				& 2 & 1.6628 &
				1.6589 & 2. & $5.71 \cdot 10^3$ &  $3.72 \cdot 10^5$ &
				1.5205 & 1.92 & $5.93 \cdot 10^3$ & $ 3.05 \cdot 10^5$ \\
				& 3 & 1.0030 &
				0.9802 & 3. & $7.37 \cdot 10^4$ &  $1.26 \cdot 10^6$ &
				0.8509 & 2.90 & $7.81 \cdot 10^4$ & $ 1.07 \cdot 10^6$ \\
				& 4 & 0.7240 &
				0.6818 & 4. & $2.86 \cdot 10^5$ &  $1.82 \cdot 10^6$ &
				0.6102 & 3.89 & $3.09 \cdot 10^5$ & $1.65 \cdot 10^6$ \\
				\hline
				\multirow{3}{*}{\textit{sources}}
				& 2 & 2.6130 &
				2.6038 & 2. & $8.16 \cdot 10^3$ &  $7.00 \cdot 10^5$ &
				2.5597 & 1.90 & $8.40 \cdot 10^3$ & $5.61 \cdot 10^5$ \\
				& 3 & 1.7667 &
				1.7223 & 3. & $1.30 \cdot 10^5$ &  $2.48 \cdot 10^6$ &
				1.6322 & 2.87 & $1.41 \cdot 10^5$ & $ 2.25 \cdot 10^6$ \\
				& 4 & 1.2764 &
				1.1848 & 4. & $5.43 \cdot 10^5$ &  $4.03 \cdot 10^6$ &
				1.0658 & 3.85 & $6.19 \cdot 10^5$ & $3.86 \cdot 10^6$ \\
				\hline
				\multirow{3}{*}{\textit{dna}}
				& 2 & 1.9220 &
				 1.9220 & 2. & $1.01 \cdot 10^2$ &  $1.02 \cdot 10^3$ &
				1.9795 & 1.97 & $1.12 \cdot 10^2$ & $1.09 \cdot 10^3$ \\
				& 3 & 1.9119 &
				 1.9119 & 3. & $3.34 \cdot 10^2$ &  $5.29 \cdot 10^3$ &
				1.9756 & 2.94 & $3.72 \cdot 10^2$ & $8.09 \cdot 10^3$ \\
				& 4 & 1.8914 &
				1.8906 & 4. & $7.25 \cdot 10^2$ &  $.68 \cdot 10^5$ &
				1.9633 & 3.86 & $8.92 \cdot 10^5$ & $1.14 \cdot 10^5$ \\
				\hline
				\multirow{3}{*}{\textit{pitches}}
				& 2 & 3.7981 &
				3.7589 & 2. & $1.02 \cdot 10^4$ &  $2.81 \cdot 10^6$ &
				3.7413 & 1.83 & $1.06 \cdot 10^4$ & $2.21 \cdot 10^6$ \\
				& 3 & 2.3500 &
				2.1569 & 3. & $2.70 \cdot 10^5$ &  $7.99 \cdot 10^6$ &
				2.0447 & 2.91 & $2.86 \cdot 10^5$ & $7.26 \cdot 10^6$ \\
				& 4 & 1.1711 &
				 0.9522 & 4. & $1.98\cdot 10^6$  &  $7.98\cdot 10^6$ &
				 0.8233 & 3.93 & $2.15\cdot 10^6$ & $7.35\cdot 10^6$ \\
				\hline
				\multirow{3}{*}{\textit{proteins}}
				& 2 & 4.0899 &
				 4.0874 & 2. & $5.62\cdot 10^2$ & $1.97\cdot 10^5$ &
				 4.2007 & 1.88 & $5.90\cdot 10^2$  & $1.97\cdot 10^5$ \\
				& 3 & 3.4890 &
				 3.2761 & 3. & $1.07\cdot 10^4$ & $9.84\cdot 10^6$ &
				 3.2574 & 2.87 & $1.12\cdot 10^4$ &  $8.18\cdot 10^6$\\
				& 4 & 1.8763 &
				 1.5245 & 4. & $1.88\cdot 10^5$ & $1.03\cdot 10^7$ &
				 1.4758 & 3.99 & $1.98\cdot 10^5$ &  $9.81\cdot 10^6$\\
				\hline
			\end{tabular}
		}
		\caption{\label{tab:first_entropy_compare}
			Entropy comparison for $H_1$ variant, entropy values are divided by $|S|$ to get bps,
			all files of size 50MB, $A$ denotes parsing generated by our algorithm,
			while $B$ is obtained by application of Theorem~\ref{thm:mean_entropy_2}
			(i.e.\ parsing minimizing $|B|H_1(B)$ among $m$ naive parsings). }
	\end{center}
\end{table}

\Section{Application --- Compression}
While the results in Table~\ref{tab:zero_entropy_compare}
and Table~\ref{tab:first_entropy_compare} show that indeed the algorithm performs better than
the naive partition,
it is hard to measure how well it can actually compress the data.
Especially because the entropy of parsing gives accurate estimation,
we still have to encode additional data and structures.

For $H_0$ to retrieve the input string $S$, we store:
Huffman compressed parsing $Y_S$, where distinct phrases are
replaced with new letters; Huffman dictionary; set of distinct phrases.
We show that the last two can be encoded efficiently;
moreover, such encoding can be easily extended to support queries,
so they will be helpful in compressed structure described later.

We start with encoding the set of phrases:
let $y = a_1 a_2 \cdots a_i \in \Sigma^{\leq m}$ be a phrase.
We treat $y$ as a~number over $|\Sigma|+1$ base,
where each letter $\sigma \in \Sigma$ is assigned a~number from $\{1,|\Sigma|+1\}$
(we do not use $0$ to avoid problematic trailing $0$'s).
Then we sort obtained numbers, call the list of such sorted numbers $P$.
We encode $P$ as a list:
$P' = \{P[1], P[2]-P[1], P[3]-P[2], \ldots, P[|\Sigma_A|] - P[|\Sigma_A|-1]\}$,
where each element of $P'$ is encoded with Elias Delta code.
Note, that on average the encoding should give similar result to
succinct encoding of Trie made of different phrases.
We also assume that the new letters in $Y_S$ correspond to order of phrases in $P$.

We now move to the Huffman dictionary.
Observe that, the Huffman dictionary can be encoded with
$2|\Sigma_A|  + |\Sigma_A| \log |\Sigma_A|$ bits,
where $2|\Sigma_A|$ bits comes from succinct encoding of a Huffman tree and 
$|\Sigma_A| \log |\Sigma_A|$ is for encoding labels in leafs of the tree.
This is also the required number of bits, on average.
Let $c_i$ be the code for a~letter $\sigma_i \in \Sigma_A$.
We append $1$ into $c_i$ and treat it as a number over binary alphabet.
Consider the list $C$ made of sorting the numbers.
We encode it as:
$C' = \{C[1], C[2]-C[1], C[3]-C[2], \ldots, C[|\Sigma_A|] - C[|\Sigma_A|-1]\}$,
with Elias Delta codes.

We also have to encode the order $L$ of letters
corresponding to codes, i.e.\ for a given code $C[i]$ we must know
which code letter $\sigma_i \in \Sigma_A$ corresponds to $C[i]$:
although the list is not sorted we still encode it
as $L[i] - L[i-1]$ with Elias Delta coding (plus additional bit per sign).
This is beneficial, as the symbols in dictionary with the same frequency
may be ordered arbitrarily, so $L$ contains monotonic sublists.

In the case of $H_1$ we build the dictionary for each $\sigma \in \Sigma_A$ separately.
The $C'$ lists are encoded separately, while the lists of letters are concatenated
and encoded together, as in the case for $H_0$.

The presented methods can be successfully applied to data compression and achieve compression ratios competitive to other compression
methods (though the results are still far behind ppmdi).
Note that at some point, when increasing $m$ the size of the dictionary grows
significantly (this is true for both $H_0$ and $H_1$ variant),
which causes the bitsize to grow with $m$,
however while increasing $|S|$ dictionary size should stay the same
for a fixed distribution of letters.

\begin{table}[h!]
	\begin{center}
		{	\setlength{\tabcolsep}{2pt}
			\renewcommand{\baselinestretch}{1}\footnotesize
			\begin{tabular}{c|c|c|c|c|ccc|ccc||c|ccc|ccc}
				File & gzip & bzip & ppmdi & \multicolumn{7}{c}{H0} & \multicolumn{7}{c}{H1}\\ 
				\hline
				&&&&&\multicolumn{3}{c}{naive} &
				\multicolumn{2}{c}{algorithm} & & &
				\multicolumn{3}{c}{naive} &
				\multicolumn{3}{c}{algorithm}\\
				\hline
				&&&&m&total&string&dict&total&string&dict&
					m&total&string&dict&total&string&dict\\
				\cline{5-18}
				\multirow{3}{*}{\textit{english}} &
				\multirow{3}{*}{3.002} &
				\multirow{3}{*}{2.272} &
				\multirow{3}{*}{1.948} &
				4 & 3.430 & 3.36 & 0.07 &
				3.237 &	3.201 &	0.04 & 
				2 & 2.749 &	2.67 &	0.08 &
				2.718 &	2.65 &	0.07 \\
				&&&&
				7 & 3.360 & 2.64 & 0.72 &
				2.817 & 2.46 & 0.36 &
				3 & 2.523 &	2.04 &	0.48 &
				2.449 &	2.00 &	0.45 \\
				&&&&
				8 & 3.557 & 2.44 & 1.11 &
				2.857 & 2.28 & 0.58 &
				4 & 2.823 &	1.60 &	1.23 &
				2.743 &	1.53 &	1.22 \\
				\hline
				\multirow{3}{*}{\textit{dblp.xml}} &
				\multirow{3}{*}{1.379} &
				\multirow{3}{*}{0.898} &
				\multirow{3}{*}{0.737} &
				4 & 3.051 & 2.97 & 0.08 &
				2.864 &	2.82 &	0.05 & 
				2 & 1.866 &	1.76 & 0.11 &
				1.743 & 1.65 &	0.10 \\
				&&&&
				7 & 2.480 & 2.00 & 0.48 &
				2.182 & 1.86 & 0.32 &
				3 & 1.479 &	1.05 & 0.42 &
				1.326 & 0.94 &	0.38 \\
				&&&&
				8 & 2.434 & 1.81 & 0.63 &
				2.134 & 1.69 & 0.44 &
				4 & 1.444 &	0.72 & 0.72 &
				1.355 & 0.66 & 0.69 \\
				\hline
				\multirow{3}{*}{\textit{sources}} &
				\multirow{3}{*}{1.863} &
				\multirow{3}{*}{1.583} &
				\multirow{3}{*}{1.337} &
				4 & 3.867 & 3.70 & 0.17 &
				3.569 &	3.49 & 0.08 & 
				2 & 2.842 &	2.63 & 0.21 &
				2.775 & 2.59 &	0.19 \\
				&&&&
				7 & 3.751 & 2.65 & 1.10 &
				3.172 & 2.51 & 0.67 &
				3 & 2.632 &	1.74 & 0.89 &
				2.523 & 1.66 & 0.86 \\
				&&&&
				8 & 3.898 & 2.40 & 1.49 &
				3.257 & 2.29 & 0.97 &
				4 & 2.852 &	1.20 & 1.65 &
				2.782 & 1.09 & 1.70 \\
				\hline
				\multirow{3}{*}{\textit{dna}} &
				\multirow{3}{*}{2.164} &
				\multirow{3}{*}{2.078} &
				\multirow{3}{*}{1.945} &
				4 & 1.948 & 1.95 & 0.00 &
				2.029 &	2.03 & 0.00 & 
				2 & 1.948 & 1.95 & 0.00 &
				1.997 & 2.00 & 0.00 \\ 
				&&&&
				7 & 1.935 &	1.93 & 0.01 &
				2.013 & 2.00 & 0.01 &
				3 & 1.924 & 1.92 & 0.00 &
				1.989 & 1.99 & 0.00\\
				&&&&
				8 & 1.945 & 1.92 & 0.03 &
				2.022 & 1.99 & 0.03 &
				4 & 1.920 & 1.90 & 0.02 &
				2.011 & 1.97 & 0.04 \\
				\hline
				\multirow{3}{*}{\textit{pitches}} &
				\multirow{3}{*}{2.686} &
				\multirow{3}{*}{2.890} &
				\multirow{3}{*}{2.439} &
				4 & 4.926 & 4.45 & 0.47 &
				4.707 &	4.43 & 0.28 & 
				2 & 4.415 & 3.77 & 0.64 &
				4.318 & 3.76 & 0.56 \\ 
				&&&&
				7 & 5.522 &	2.99 & 2.53 &
				5.175 & 2.97 & 2.20 &
				3 & 4.986 & 2.17 & 2.82 &
				4.797 & 2.06 & 2.74\\
				&&&&
				8 & 5.570 & 2.63 & 2.94 &
				5.221 & 2.60 & 2.62 &
				4 & 5.102 & 0.97 & 4.14 &
				4.851 & 0.84 & 4.01 \\
				\hline
				\multirow{3}{*}{\textit{proteins}} &
				\multirow{3}{*}{3.791} &
				\multirow{3}{*}{3.645} &
				\multirow{3}{*}{3.364} &
				4 & 4.197 & 4.14 & 0.06 &
				4.268 &	4.22 & 0.05 & 
				2 & 4.165 & 4.10 & 0.06 &
				4.281 & 4.22 & 0.07 \\ 
				&&&&
				7 & 5.171 &	3.14 & 2.03 &
				5.028 & 3.14 & 1.88 &
				3 & 5.100 & 3.29 & 1.81 &
				4.966 & 3.27 & 1.70\\
				&&&&
				8 & 5.170 & 2.74 & 2.43 &
				5.082 & 2.73 & 2.35 &
				4 & 5.344 & 1.54 & 3.81 &
				5.198 & 1.49 & 3.71 \\
				\hline

			\end{tabular}
		}
		\caption{\label{tab:compression_compare} Compression results, values in bps, every file is 50MB, \textit{string} and \textit{dict} denote the size of encoding of string
		and dictionary respectively, \textit{algorithm} --- parsing generated by our algorithm ($H_0/H_1$),
		\textit{naive} --- parsing minimizing $|B|H_0(B)/|B|H_1(B)$ among $m$ naive parsings. }
	\end{center}
\end{table}

\Section{Application --- Structure}
We now show how to construct a structure which allows random access.
The high-level idea of previous solutions
(e.g.~\cite{FerraginaV07SimpStat,GonzalezNStatistical}) was
to encode the parsing with entropy coding (e.g.\ Huffman),
store set of phrases in array indexed by codes (i.e.\ $A[c_i] = w, w \in\Sigma_S^m$),
and store additional structure which is able to retrieve $i$-th
encoded code in entropy coded bitstring
(this can be done with $\Ocomp(|Y_S|\log \log |S|$ bits).
Then the letter $S[i]$ can be easily retrieved in constant time,
assuming that we can read the code
in constant time: as phrases are equal, letter $S[i]$ will be in $i/m$ block.

The case for $H_1$ is similar~\cite{GrossiDynamicIndexes},
but we store $|\Sigma_{Y_S}|$ dictionaries, and every $l$-th phrase is stored
explicitly (note that to decode the phrase we have to have the previous one).
Hence the decoding starts at explicitly stored phrases,
and decodes at most $l-1$ phrases to the right,
thus decoding takes $\Ocomp(l)$ time.

It turns out that both the above structures can be modified to support
parsing returned by our algorithms:
the only difficulty is that we do not have equal-length phrases, hence when queried for $S[i]$ we do not know which code to return.
This can be solved by using succinct partial sum structure on length of phrases:
queried for $S[i]$, we know that $i$ is in phrase $j$ such that
$\sum_{k=1}^{j-1} |y_k| < i \leq \sum_{k=1}^j |y_k|$.
Such structure uses $|Y_S|\log m + o(|Y_S|)$ bits, which is $\Ocomp(|S|\log \log / \log_\sigma |S|)$ for $|Y_S| = \Ocomp(|S|/\log_\sigma|S|)$ and $m = \log_\sigma |S|$
(in this case we get the same redundancy
as for structures from~\cite{FerraginaV07SimpStat,GonzalezNStatistical,GrossiDynamicIndexes}).

We also use more practical encoding of dictionary:
we use succinct partial sums for sequences $P'$ and $C'$,
this allows to answer queries in $\Ocomp(1)$ time.
To store the sequence of letters which corresponds to codes, $L$,
we use Elias Delta compressed array.
Such encoding clearly gives much better result than storing all of the possible phrases explicitly.

We implemented the structures for $H_0$,
the results of experiments are in Table~\ref{tab:structure_h0_compare}.
The increased $T_r$ and $\delta_s$ for structure for our parsings
is due to the need for additional structure,
our implementation supports the tradeoff between those two.
\begin{table}[h!]
	\begin{center}
		{
			\renewcommand{\baselinestretch}{1}\footnotesize
			\begin{tabular}{cc|ccc|cccc|cccc}
				file & m &
				\multicolumn{3}{c}{uncompressed} &
				\multicolumn{4}{c}{naive} &
				\multicolumn{4}{c}{algorithm} \\
				\hline
				 &&
				 bps & $T_r$ & $T_b$ & bps & $\delta_s$ & $T_r$ & $T_b$ & bps & $\delta_s$ & $T_r$ & $T_b$ \\
				\hline
				  \multirow{2}{*}{\textit{english}}  & 7 &
				  \multirow{2}{*}{8} & \multirow{2}{*}{0.008} & \multirow{2}{*}{0.003} & 
				  4.205 & 0.845 & 3.14 & 16.06 &
				  3.911 & 1.094 & 14.38 &  18.88 \\
				  & 8  &  &  &  &
				  4.332 & 0.775 & 3.16 & 14.36 &
		 		  3.870 & 1.013 & 14.61 &  16.82 \\
				 \hline
				 \multirow{2}{*}{\textit{dblp}} & 7 &
				 \multirow{2}{*}{8} & \multirow{2}{*}{0.008} & \multirow{2}{*}{0.003} &
				 3.215 & 0.735 & 3.20 & 16.49 & 3.174 & 0.992 & 12.76 & 16.31 \\
				 & 8  & & & & 3.101 & 0.668 & 3.20  & 14.55 &
				 3.006 & 0.872 & 13.37 & 15.12 \\
				 \hline
				 \multirow{2}{*}{\textit{sources}} & 7  &
				 \multirow{2}{*}{8} & \multirow{2}{*}{0.008} & \multirow{2}{*}{0.003} &
				 4.692 & 0.940 & 3.18 & 15.94 &
				 4.308 & 1.136 & 14.26 & 18.11 \\
				  & 8  & & & &
				 4.743 & 0.845 & 3.24 & 14.70 &
				 4.293 & 1.036 & 16.63 & 16.47 \\
				 \hline
				 \multirow{2}{*}{\textit{dna}} & 7  &
				 \multirow{2}{*}{8} & \multirow{2}{*}{0.009} & \multirow{2}{*}{0.003} &
				 2.573 & 0.637 & 2.90 & 14.92 &
				  2.926 & 0.913 & 12.36 & 15.91 \\
				 & 8  &  & & & 2.582 & 0.637 & 3.02  & 12.28 &
				 2.874 & 0.851 & 12.22 & 13.84 \\
				\hline
			\end{tabular}
		}
		\caption{\label{tab:structure_h0_compare} Structure for $H_0$,
		comparison of bps/time[sec] for operations,
		$\delta_s$ --- difference between bps of compressed file (using our encoding) and bps of queryable structure,
		$T_r$ --- read time for a random list of $10^6$ letters,
		$T_b$ --- read time for a read of $10^3$ blocks of 50KB. Ran on Intel i5-7400.
		All files are of size 50MB.}
	\end{center}
\end{table}
\vspace{-5pt}

\Section{References}
\bibliographystyle{IEEEbib}
\bibliography{references}
\appendix

\Section{Appendix}
\subsection{Proofs}
\begin{proof}[Proof of Lemma~\ref{lem:alg_zero_bound}]
	Consider the function $p'(y) = p_1(a_1) \cdot p_1(a_2)\cdots p_1(a_j)$
	i.e.\ the function which does not take into the account phrases' lengths.
	Proof of Theorem~\ref{thm:mean_entropy_2}~\cite{entropyBoundsUnpublished} states that there exist a parsing $Y'_S$
	satisfying: 
	$$
	\sum_{y \in Y'_S} \log p'(y) \leq \frac{|S|}{m} \sum_{0 \leq i \leq m-1} H_i(S) \enspace .
	$$
	The algorithm finds a parsing $Y_S$ which minimizes
	$$\sum_{y \in Y_S} \log p(y) =
	\sum_{y \in Y_S} \log p'(y) + |Y_S|\log m \enspace .$$
	Hence, $$\sum_{y \in Y_S} \log p(y) - |Y_S|\log m \leq  \frac{|S|}{m} \sum_{0 \leq i \leq m-1} H_i(S) \enspace .$$
	
	Now we argue that for all substrings of length at most $m$ the sum of values $p(\cdot)$
	summed over different substrings is at most $1$.
	For a fixed string $S$ over alphabet $\Sigma$ and fixed length $z$
	it holds: $\sum_{w \in \Sigma^z} p'(w) \leq 1$, i.e.\ $p'(\cdot)$ for all possible strings of length $z$
	sums up to at most $1$~\cite[proof of Theorem 7]{entropyBoundsUnpublished}).
	As there are at most $m$ different lengths, the claim holds.
	Thus by applying Lemma~\ref{lemma_prob} the first inequality of the Lemma holds.
	The second inequality follows from the first one and the fact that
	$\log \sigma \geq H_0(S) \geq H_1(S) \geq \cdots \geq H_i(S)$ for every $i$.	
\end{proof}

\begin{proof}[sketch of proof of Theorem~\ref{thm:thm2}]
	We define $p(\cdot)$  values for each phrase $y= a_1 a_2 \cdots a_j$
	preceded by phrase $y'= b_1 b_2 \cdots b_h$ (so $y'y$ is a substring of $S$):
	$$
	p(y, y') = p_l(|y|) \cdot p_1(a_1) \cdot p_2(a_2)\cdot \cdots \cdot p_j(a_j)$$
	where $p_l(|y|) = \frac{\cnt{|y|}}{|Y_S|}$
	and
	\begin{equation*}
	p_i(a_i) = 
	\begin{cases}
	\frac{\cnt{ b_{ h-(i-k)+1 } \cdots b_{h-1} b_h a_{1}\cdots a_i}}{\cnt{ b_{ h-(i-k)+1 } \cdots  b_{h-1} b_h a_{1}\cdots a_{i-1} }}, & \text{for} \  i \leq k;\\
	\frac{\cnt{a_{i-k}\cdots a_i}}{\cnt{a_{i-k}\cdots a_{i-1}}}, & \text{for} \  i > k.\\
	\end{cases}
	\end{equation*}
	In short, every letter is assigned the empirical probability
	of occurring in the $k$-th letter context preceding this letter.
	
	By definition, for a string $T \in \Sigma^*$,
	$|T|H_1(T) =  \sum_{T_\sigma, \sigma \in \Sigma} |T_\sigma| H_0(T_\sigma)$,
	where $T_\sigma$ is a string obtained by concatenating all the letters of $T$
	which occur in one-letter context $\sigma$ (e.g.\ for $T=abacaac$, $T_a = bcac$).
	
	Now the proof is straightforward:
	for a parsing $Y_S$ we derive its alphabet $\Sigma_{Y_S}$,
	constructs strings $T_{\sigma'}$ for $\sigma' \in \Sigma_{Y_S}$,
	and apply Lemma~\ref{lemma_prob} for each $T_{\sigma'}$.
	Observe that as each letter is assigned its empirical probability of occurring in $k$-th letter context preceding this letter,
	by definition of $H_k(S)$ all the values $-\log p(y)$ summed over all factors
	sum up to the claimed bounds, i.e.:
	$$
	\sum_{\sigma' \in \Sigma_{Y_S}} \sum_{y \in T_\sigma'} \log p(y,\sigma') 
	 = |S|H_k(S) + |L|H_0(L) \enspace ,
	$$
	where $|L|H_0(L)$ is the entropy of lengths.
	To apply Lemma~\ref{lemma_prob} we need the fact that
	$\sum_{\sigma'' \in \Sigma_{T_\sigma}} p(\sigma'') \leq 1$,
	where $\Sigma_{T_\sigma}$ is the alphabet of $T_\sigma$
	(we need this for each $T_\sigma$),
	yet it can be shown analogously as the claim used in
	\cite[proof of Theorem~\ref{thm:thm1}]{entropyBoundsUnpublished}.
\end{proof}

\begin{proof}[sketch of proof of Theorem~\ref{thm:mean_entropy_2}]
	We define $p(\cdot)$  values for each phrase $y= a_1 a_2 \cdots a_j$
	preceded by phrase $y'$:
	$$
	p(y, y') = p_l(|y|) \cdot p_1(a_1) \cdot p_2(a_2)\cdot \cdots \cdot p_j(a_j)$$
	where
	$p_i(a_i) = \frac{ \cnt{y' a_{1}\cdots a_i}}{ \cnt{y' a_{1}\cdots a_{i-1}}}$
	and $p_l(|y|) = \frac{\cnt{|y|}}{|Y_S|}$.
	
	This should be understood as follows:
	if the previous factor has length $|y'|$,
	then we look at the empirical probability that 
	$i$-th letter of factor $y$ occurs in $S$ in an $|y'| + i-1$ letter context.
	Or equivalently, we assign the letter $a_i$ the cost, i.e.\ $-\log p(a_i)$,
	which roughly corresponds to cost of encoding of $a_i$ with $|y'| + i-1$-order
	entropy coder.
	
	As in the proof of Theorem~\ref{thm:thm2}, we have:
	for a string $T \in \Sigma^*$,
	$|T|H_1(T) =  \sum_{T_\sigma, \sigma \in \Sigma} |T_\sigma| H_0(T_\sigma)$,
	where $T_\sigma$ is a string made by concatenating all the letters of~$T$.
	
	Again we must show that $\sum_{\sigma' \in \Sigma_{T_\sigma}} p(\sigma') \leq 1$,
	so we can apply Lemma~\ref{lemma_prob}.
	This can be shown in exactly the same way as in proof of Theorem~\ref{thm:mean_entropy_2},
	see~\cite{entropyBoundsUnpublished}.
	
	It was left to show the parsing $Y_S$ for which
	$\sum_{y \in Y_S} -\log p(y) \leq \frac{|S|}{l} \sum_{i=l}^{2 l-1} H_i(S) + \Ocomp(\log |S|)$.
	We use similar arguments as in proof of Theorem~\ref{thm:mean_entropy_2}:
	we look at $l$ possible parsings where each phrase is of length $l$,
	except for the first and the last one, which can be shorter.
	Summing up all the values $\sum_{y \in Y_S} -\log p(y, y')$ for the $l$ parsings 
	we end up with
	$$
	|S|\sum_{i=l}^{2 l-1} H_i(S) + \sum_{i<l}\sum_{y \in Y_S^i} p_l(|y|) .
	$$ 
	Note that the second term is the sum of entropies of lengths for each parsing,
	i.e.\ the~entropy of strings made of lengths of phrases.
	Note that the $ \sum_{i<l}\sum_{y \in Y_S^i} p_l(|y|) = \Ocomp(l \log |S|)$,
	as $\sum_{y \in Y_S^i} p_l(|y|) = \Ocomp(\log |S|)$ due to the fact that
	only the first and last phrases can have different lengths.
	Hence we can conclude that for at least one of those parsings it holds:
	$$
	\sum_{y \in Y_S} -\log p(y) \leq \frac{|S|}{l} \sum_{i=l}^{2 l-1} H_i(S) + \Ocomp(\log |S|),
	$$
	which yields the claim.
\end{proof}

\begin{proof}[sketch of proof of Lemma~\ref{lem:algo_first}]
Analogously as in the proof of Lemma~\ref{lem:algo_zero},
there exist a~parsing (by the proof of Theorem~\ref{thm:mean_entropy_2}) where 
\begin{equation}
\label{eq:lem_alg_first}
\sum_{y \in Y'_S} \log p'(y', y) \leq \frac{|S|}{m} \sum_{0 \leq i \leq m-1} H_i(S) \enspace ,
\end{equation}
where $p'(y', y) = p(y', y)/p_l(|y|)$.
As our algorithm finds a parsing $Y_S$ which minimizes the $\sum_{y \in Y'_S} \log p'(y', y)$
the~\eqref{eq:lem_alg_first} must hold for $Y_S$.
Hence, by repeating the reasoning from the proof of Theorem~\ref{thm:mean_entropy_2},
we get the first inequality.
Again, the second inequality follows from the fact that
$\log \sigma \geq H_0(S) \geq H_1(S) \geq  \cdots \geq  H_i(S)$ for every $i$.
\end{proof}

\subsection{Implementation details}
To realize the succinct sums we use $rrr$-vector from sdsl library~\cite{gbmp2014sea}.
With the use of $rrr$-vector we encode:
Huffman codes lengths for each phrase of parsing,
(we subtract minimum code length before applying this structure);
and sorted Huffman codes ($C$).

We use (difference) delta-encoded vectors with random access from sdsl library (enc\_vector),
different phrases ($P$) and list $L$.

For structure on lengths of phrases (which allows to get phrase containing $i$-th letter)
we develop our own structure: we build an array of $|S|/d$ elements $Z$
where $Z[j]$ is the index of factor containing $j$-th letter,
similarly we store offsets $O[j]$, which gives position of $j$-th letter in factor.
$Z$ is encoded with enc\_vector, $O$ is an array of bit-packed integers of $\log m$ bits.
Note that instead of storing smaller structure for $j$ values not being multiplies of $d$
we can just decode a few factors to the right---thus we have space/time tradeoff for a parameter $d$.
Observe that setting $d = \Theta(m)$ gives a solution which
on average, have constant query time
(when average factor length is $\Ocomp(m)$ which, as the experiments show, is the case in practice).

The bottleneck of our solution is the $rrr$-vector structure,
thus improving this part should give much better query times.

On the side note one could save additional space by encoding factors of different lengths separately,
as now it seems like we are storing the information on lengths of factors twice
(both in a entropy of parsing and in external structure).

\end{document}